\documentclass{article}
\usepackage{a4wide}

\usepackage{amsthm,amsmath,amsfonts,amssymb}
\usepackage{xcolor}
\usepackage{xspace}
\usepackage{graphicx}
\usepackage{hyperref}

\newtheorem{remark}{Remark}

\newtheorem{lemma}{Lemma}

\newcommand{\bydef}{:=}

 \newcommand{\beff}{{\ensuremath{\beta_{\mathrm{eff}}}}}
 \newcommand{\bB}{\mathbb{B}}
 \newcommand{\bC}{\mathbb{C}}
 \newcommand{\bD}{\mathbb{D}}
 \newcommand{\bG}{\mathbb{G}}
 \newcommand{\bx}{\mathbf{x}}
 \newcommand{\efs}{\ensuremath{\mathsf{EFS}}}
 
\DeclareMathOperator*{\argmax}{arg\,max}
\usepackage{mathtools}

 \newcommand{\bA}{\ensuremath{\mathbf{A}}}
 
 \newcommand{\Suc}{\ensuremath{\mathsf{S}}\xspace}
 \newcommand{\Inf}{\ensuremath{\mathsf{I}}\xspace}
 \newcommand{\Rem}{\ensuremath{\mathsf{R}}\xspace}
 
 \usepackage{authblk}

\begin{document}

\title{Group centrality in optimal and suboptimal vaccination for epidemic models in contact networks.}


 \author[1]{J. Orestes Cerdeira\footnote{jo.cerdeira@fct.unl.pt, ORCID-ID:0000-0002-3814-7660}}

 \author[1,2]{Fabio A. C. C. Chalub\footnote{Corresponding author: facc@fct.unl.pt, ORCID-ID: 0000-0002-8081-9221}}

 \author[1]{Matheus Hansen\footnote{mh.francisco@fct.unl.pt, ORCID-ID: 0000-0003-0125-9033}}

%
 \affil[1]{Center for Mathematics and Applications (NOVA Math), NOVA School of Science and Technology, Universidade NOVA de Lisboa, Quinta da Torre, 2829-516, Caparica, Portugal.}
 \affil[2]{Department of Mathematics, NOVA School of Science and Technology, Universidade NOVA de Lisboa, Quinta da Torre, 2829-516, Caparica, Portugal.}

 \date{\today}

 \maketitle




\begin{abstract}
The pursuit of strategies that minimize the number of individuals needing vaccination to control an outbreak is a well-established area of study in mathematical epidemiology. However, for certain diseases, public policy tends to prioritize immunizing vulnerable individuals over epidemic control. As a result, optimal vaccination strategies may not always be effective in supporting real-world public policies. A similar situation happens when a new vaccine is introduced and is in short supply, as target priority groups for vaccination have to be defined. In this work, we focus on a disease that results in long-term immunity and spreads through a heterogeneous population, represented by a contact network. We study four well-known group centrality measures and show that the GED-Walk offers a reliable means of estimating the impact of vaccinating specific groups of indi\-vi\-duals, even in suboptimal cases. Additionally, we depart from the search for target individuals to be vaccinated and provide proxies for identifying optimal groups for vaccination. While the GED-Walk is the most useful centrality measure for suboptimal cases, the betweenness (a related, but different centrality measure) stands out when looking for optimal groups. This indicates that optimal vaccination is not concerned with breaking the largest number of transmission routes, but interrupting geodesic ones.
\end{abstract}

\noindent\textbf{MSC-class}: 05C90, 90C35, 92D30\\
\textbf{ACM-class}: G.2.2

\section{Introduction}

Protecting vulnerable groups and efficiently controlling disease transmission may be challenging to achieve simultaneously. For certain diseases such as flu and COVID-19, young healthy adults tend to have more contacts that could potentially transmit the disease compared to elderly individuals; however, the disease manifestation tends to be milder in the former group than in the latter one. In practice, while most mathematical models focus on efficient disease control~\cite{Wang_2016,Holme_2017,Holme_2017b,Hughes_IEEE_2020,Doutoretal_JMB2016}, for certain infectious diseases governments prioritize socially relevant groups over disease control. That may particularly be true in the initial phases of vaccination campaigns, when vaccines are scarce, cf.~\cite{GreenBook_2020,Rashi_2023} for COVID-19. The difficult question of defining priorities when vaccines are scarce is again relevant in the more recent outbreak of mpox in the Democratic Republic of Congo~\cite{Adepoju2024}, see also~\cite{Shibadas2019,Nguyen2025} for the case of a newly introduced vaccine against Dengue fever, in which vaccination priority is given to teenagers.

In this work, we develop a framework that allows us to study the impact on the disease dynamics of the vaccination of any given group of individuals. In particular, assuming that the population is organized in a contact network, we will provide proxies based solely on the topology of the network to estimate the outbreak size as a function of the group of vaccinated individuals. In the next step, we will use heuristics to optimize that group. These proxies and heuristics will be studied assuming simple epidemic dynamics and will be based solely on the population structure. In the sequel, we will validate our result by studying directly the epidemic dynamics and showing that a suitable choice of the proxy will provide a group close to the one providing optimum disease control, in this last case, obtained directly from numerical simulations. 

We start this work by exploring the well-known SIR (susceptible, infectious, removed/ recovered) model, introduced in~\cite{Ross_1916, Kermack_McKendrick_1927} for unstructured populations. Here, on the contrary, the population is structured by a contact network represented by an undirected graph $G=(V,E)$. The set of vertices (or nodes) $V$ represents individuals in the population, and the set of edges (or links) $E$ represents potential transmission contacts.
Each individual in the population is in one and only one of the three states: \textbf{S}usceptible, \textbf{I}nfectious, and \textbf{R}ecovered. The transitions between these states are defined as: $S+I\stackrel{\beta}{\to}2I$, $I\stackrel{\gamma}{\to}R$, where $\beta$ (transmission rate) and $\gamma$ (recovery rate) are strictly positive parameters. For more information on compartmental models in epidemiology, including the SIR model and its variations, cf.~\cite{anderson1991infectious, Thieme_2003}.

Notable examples of diseases described by compartmental models in undirected contact networks include influenza~\cite{Martin_2011}, Ebola~\cite{Rizzo_2016}, and COVID-19~\cite{Nizami_2023}. The transmission of non-human diseases can be described in a similar manner, as observed in primates~\cite{Romano_2016}, dolphins~\cite{Leu_2020}, and raccoons~\cite{Hirsh_2013}. For a comprehensive description of epidemic models in networks, see the reviews~\cite{Keeling_2005,Pastor_2015} and the reference work~\cite{kiss2017mathematics}.

One of the main inspirations for the present work comes from the Influence Maximization Problem (IMP), first introduced in a study of peer-to-peer marketing in a population structured by a contact network~\cite{Domingos_Richardson_2001}. The main objective of IMP is to identify influential individuals within a network as potential targets for advertising. IMP is a diffusion model where there is a set of initially active nodes, and each active node has the potential to influence its neighbors, thereby potentially activating them. Two specific realizations of IMP are the Linear Threshold Diffusion Model (LTDM) and the Independent Cascade Model (ICM). In LTDM, a node becomes activated with a probability that depends on the number of its active neighbors, while in ICM, each active neighbor independently attempts to activate an inactive node, with only one opportunity to activate each neighbor. When the infection rate is low, both formulations are equivalent~\cite{Holme_2017}. IMP aims to achieve efficient spreading, seeking a group of nodes $S^*$ with a given cardinality that maximizes the final set of active nodes $\sigma(S^*)$. For further formulations of IMP and results in this context, please refer to~\cite{LiFanWang_2018,Solanki2025,Jaouadi_2024}.

The current work has two main objectives: i) to estimate the final number of contaminated individuals as a function of the initially vaccinated ones; and ii) to find optimal subsets to be vaccinated.

The first objective is independent of the second but has the latter as a particular case. In a realistic vaccination campaign, the goal is not only to minimize the total number of contaminated individuals during an outbreak but also to protect specific subpopulations. Therefore, it is important to estimate the effects of vaccination on the dynamics, even for groups that may not be optimal. As a corollary, it will help to find optimal subsets. In simpler terms, {this means finding a set $V^*$ that, when removed from the graph, minimizes $\sigma[V\backslash V^*](S^*)$ for a given fixed initial set of active sites $S^*$.  

Even though the second problem has attracted most of the attention of modelers working on the vaccinations problem, the first problem has a larger potential to assist decision-makers. 

Many developments in the study of outbreaks in networks were inspired by the IMP problem described above. Until recently, the problem of identifying optimal sets $S^*$ of \emph{spreaders} and optimal sets $V^*$ of \emph{blockers} has been thought to be closely connected~\cite{Namtirtha_2021,Holme_2017b,Bucur_2020,Wei_2022}. Assuming fixed cardinality, an optimal set of \emph{spreaders} is given by a set of individuals that being initially infected will maximize the final number of infected individuals; an optimal set of \emph{blockers} is a set of individuals that after being vaccinated will minimize the final number of infected individuals, assuming an initial condition given by one infectious individual in a random node.

In fact, according to \cite{Cheng_Yong-Hong_2020}, there is an equivalence between finding $S^*$ and $V^*$ in the LTDM. However, this may not hold in the ICM~\cite{Radicchi_PhysRevE}, as we will corroborate with precise statements in Rmk~\ref{rmk:spreaders_blockers}. We note that in epidemic models, contagion rules share similarities with the ICM, but not with the LTDM, as each contact between individuals potentially independently transmits the disease.

Furthermore, determining if a given individual will become eventually infectious is an NP-complete problem~\cite{Shapiro_2012,Kempe_Tardos_2003}. Additionally, identifying the most efficient spreader $S^*$ is NP-hard~\cite{Habiba_2011}. Even taking into consideration the non-equivalence discussed in the previous paragraph, it is reasonable to expect that the identification of $V^*$ will be time-consuming from a computational perspective. Therefore, finding topological proxies for the dynamics is important.

Network centrality, or, simply `centrality', is a key concept in network studies and is increasingly important in understanding epidemics. In simple terms, centrality measures the importance of a node or group of nodes (denoted as $S$) to the entire network. There are various centrality measures, some of which may or may not be significant in specific contexts. Some measures are local (such as the number of first neighbors of $S$), while others are non-local (such as the number of geodesics between all possible pairs of nodes that pass through $S$). The former is easier to compute but has strong limitations in its applicability, while the latter provides a more representative view of the relevance of $S$ in the network dynamics but is more challenging to compute. For more information, the interested reader should refer to~\cite{Newman_2010} for a discussion of these properties and others related to networks.

Models for the identification of node centrality and their application in the study of disease spread are referenced in~\cite{Arruda_2014,Dekker2013NetworkCA,Wei_2022,Holme_2017b,Bucur_2020,Berahmand_2018,Namtirtha_2021}. The concept of target group identification for the IMP using group centrality is discussed in~\cite{Ganguly_2024}. However, to the best of our knowledge, no previous work considers both disease dynamics in a network and the search for target group identification using group centrality measures simultaneously.

The outline of this work is as follows: in Sec.~\ref{sec:basic}, we introduce the SIR model in networks and the most relevant group centrality concepts. In Sec.~\ref{sec:exact}, we present some simple graphs in which exact calculations are possible and compare the epidemic final size (EFS) with several distinct concepts of group centrality. In Sec.~\ref{sec:numerical}, we study classical graphs in which exact results are not available and numerically compare the EFS with group centrality. We also show that centrality measures are valuable for estimating
the EFS of arbitrary groups. In Sec.~\ref{sec:annealing}, we look for optimal groups with respect to given centrality measures and use the simulated annealing algorithm to show that vaccination strategies based on group centrality result in an EFS that is nearly optimal. In Sec.~\ref{sec:USMap}, we validate our ideas using a real-world network which, although not based on individual interactions, helps illustrate some of our key findings and limitations. We conclude with a discussion of the merits and limitations of our proposal in the discussion section, Sec.~\ref{sec:discussion}.

\section{Basic notation}
\label{sec:basic}

Let $G=(V,E)$ be an undirected \emph{graph} where $V=\{1,2,3,\dots,N\}$ is a finite set of \emph{vertices}, and $E\subset \{\{i,j\}| i\ne j, i,j\in V\}$ 
is the set of \emph{edges}. Let $\bA=\left(a_{ij}\right)_{i,j\in V}\in\{0,1\}^{N^2}$, where $a_{ij}=1$ if $\{i,j\}\in E$ and $a_{ij}=0$ otherwise, be the \emph{adjacency matrix} of $G$.

We will use the terms \emph{networks} and \emph{nodes} to refer to graphs and vertices, respectively. It is also common to call the edges as \emph{links}.

A \emph{walk} of size $k$ between nodes $i$ and $j$ is a sequence of nodes $\{i_0,i_1,\dots,i_k\}\subset V$, such that $i_0=i$, $i_k=j$ and $a_{i_li_{l+1}}=1$, for $l=0,\dots,k-1$. A \emph{path} is a walk without repeated nodes. 
A network is \emph{connected} if there is a path connecting every pair of nodes. 
The \emph{distance} $d_{ij}$ between nodes $i$ and $j$ is the minimum size of a path from node $i$ to node $j$; all paths with size $d_{ij}$ are called \emph{geodesic paths}. 

In the study of disease dynamics in networks, each node represents an individual, and each edge is a contact between individuals possibly transmitting the disease. Each node $i$ is in one of the three following states: $x_i=\Suc$, $x_i=\Inf$, or $x_i=\Rem$, representing \emph{susceptible}, \emph{infectious} or 
\emph{recovered} individuals, respectively, the base of the SIR model. For all practical purposes, vaccinated individuals are equivalent to recovered individuals.
 This is frequently called the SIR/V model. States are time-dependent, therefore, whenever necessary we will call $x_i(t)$ the state of node $i$ at time $t$.

Initial conditions, except with otherwise stated, will be given by $x_{i_*}(0)=\Inf$ and $x_i(0)=\Suc$ for all $i\ne i_*\in V$. Node $i_*$ is the \emph{patient zero} of the outbreak.

The SIR model is defined by two positive constants, $\beta,\gamma\in(0,1]$. States are updated according to the following rules, which shall be applied simultaneously to all $i\in V$.
\begin{enumerate}
    \item If $x_i(t)=\Suc$, then $x_i(t+1)=\Suc$ with probability $(1-\beta)^{\kappa_i}$. The exponent $\kappa_i=\sum_ja_{ij}\delta_j$, where $\delta_j=1$ if $x_j(t)=\Inf$ and $0$ otherwise, is the number of first neighbors of $i$ at state \Inf at time $t$. With complementary probability, $x_i(t+1)=\Inf$.
    \item If $x_i(t)=\Inf$, then $x_i(t+1)=\Rem$ with probability $\gamma$ and $x_i(t+1)=\Inf$ with complementary probability.
    \item If $x_i(t)=\Rem$, then $x_i(t+1)=\Rem$ with probability 1.
\end{enumerate}
Note that the update from $\bx(t)\to\bx(t+1)$, where $\bx=(x_i)_{i\in V}$ is proceed only after all nodes $i\in V$ were analysed. In this case, it is common to say that the update is \emph{syncronous}.

After a sufficiently long time, all nodes will be \Suc or \Rem. The system reaches a stationary state and the outbreak finishes. The cardinality of the final set of \Rem nodes is the \emph{epidemic final size} (EFS).

\begin{remark}\label{rmk:betagamma}
  In standard SIR models based on differential equations, the parameters $\beta$ and $\gamma$ are defined differently. In this context, $\beta$ represents the rate at which susceptible individuals become infected per infectious individual, per unit time, while $\gamma$ indicates the rate at which individuals recovers per unit of time. As a result, both $\beta$ and $\gamma$ can take on any non-negative value~\cite{Thieme_2003}.

In this work, however, we define $\beta$ and $\gamma$ as transition probabilities. Specifically, \(\beta\) represents the probability of susceptible individuals transitioning to the Infectious class based on the number of infectious neighbors per unit of time, while $\gamma$ indicates the probability of individuals transitioning out of the Infectious class per unit of time. Consequently, the values for $\beta$ and $\gamma$ are constrained to the interval $[0, 1]$. We do not consider cases where either parameter equals zero, as this would lead to a trivial model.
\end{remark}

In the next lemma, we introduce a simplification that will help analytical and numerical calculations:
\begin{lemma}\label{lem:rescaling}
Let $G=(V,E)$ be a network, and consider the SIR model with parameters $(\beta,\gamma)$. Then, the EFS is identical to the EFS of the SIR model with parameters $(\beff,1)$ with
\[
\beff\bydef\frac{\beta}{1-(1-\beta)(1-\gamma)}\in(0,1]\ .
\]
\end{lemma}

\begin{proof}
Consider a given edge between a node $i$, at state \Inf, and an adjacent node $j$ at state \Suc. Eventually, the pair $(i,j)$ will be in one of the following three states: $(\Rem,\Suc)$, $(\Rem,\Inf)$, or $(\Inf,\Inf)$. Furthermore, $\lim_{t\to\infty}x_i(t)=\Rem$.

The probability of finding the pair $(i,j)$ at state $(\Rem,\Suc)$ is given by the probability that in the first time-step $i$ does not infect $j$ and recovers, $(1-\beta)\gamma$, plus the probability that it does not infect nor recover in the first time-step, does not infect and recover in the second time-step, etc. Therefore, the probability that $i$ does not infect $j$ will be given by
\[
(1-\beta)\gamma+(1-\beta)(1-\gamma)(1-\beta)\gamma+\dots=(1-\beta)\gamma\sum_{k=0}^\infty(1-\beta)^k(1-\gamma)^k=\frac{(1-\beta)\gamma}{1-(1-\beta)(1-\gamma)}\ .
\]
The convergence of the above series is guaranteed from the fact that $0<\beta,\gamma<1$ implies that $(1-\beta)(1-\gamma)<1$, and the resulting sum is a nonnegative number strictly bounded by 1.
The probability $\beff$ that $i$ eventually infects $j$ is the complementary probability and is, therefore, a positive number bounded by 1. Rescaling time such that this happens in one interaction, we have $\gamma_{\mathrm{eff}}=1$.
\end{proof}

See Table~\ref{tab:beff} for a pictorial proof of the above Lemma.

From now on, except if otherwise stated, we will drop the subindex '$\mathrm{eff}$' and will consider $\gamma=1$.

\begin{table}
\centering
\begin{tabular}{c|c}
    Transitions & Probabiity\\
    \hline
\phantom{\rule{0cm}{0.5cm}}\raisebox{-0.1cm}{\includegraphics[width=0.1\textwidth]{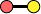} \raisebox{0.1cm}{\mbox{${\large\boldsymbol{\Rightarrow}}$}} \includegraphics[width=0.1\textwidth]{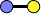}} & $(1-\beta)\gamma$\\[0.1cm]
\raisebox{-0.1cm}{\includegraphics[width=0.1\textwidth]{Table1_figA} \raisebox{0.1cm}{\mbox{${\large\boldsymbol{\Rightarrow}}$}} \includegraphics[width=0.1\textwidth]{Table1_figA} \raisebox{0.1cm}{\mbox{${\large\boldsymbol{\Rightarrow}}$}} \includegraphics[width=0.1\textwidth]{Table1_figB}} & $(1-\beta)^2(1-\gamma)\gamma$\\[0.1cm]
    \raisebox{-0.1cm}{\includegraphics[width=0.1\textwidth]{Table1_figA} \raisebox{0.1cm}{\mbox{${\large\boldsymbol{\Rightarrow}}$}}
    \includegraphics[width=0.1\textwidth]{Table1_figA} \raisebox{0.1cm}{\mbox{${\large\boldsymbol{\Rightarrow}}$}}
    \includegraphics[width=0.1\textwidth]{Table1_figA} \raisebox{0.1cm}{\mbox{${\large\boldsymbol{\Rightarrow}}$}}
    \includegraphics[width=0.1\textwidth]{Table1_figB}} & $(1-\beta)^3(1-\gamma)^2\gamma$\\
$\textbf{{\large\vdots}}$ & $\textbf{{\large\vdots}}$
\end{tabular}
\caption{Consider an adjoint pair of nodes, initially at state (\Inf,\Suc), marked by red and yellow, respectively. After a sufficiently long time, the infectious one will become \Rem (blue), while the susceptible one will be either infectious (and then recovered), counting for the EFS, or will remain susceptible, without being infected. The table shows all possible states in time such that the second node is never infected, up to a final state (\Rem,\Suc) in which the link becomes irrelevant. The right column shows the probability that the transition (\Inf,\Suc)$\Rightarrow$(\Rem,\Suc) happens in exactly $\kappa$ units of time (interactions). In this case, in the first $\kappa-1$ interactions the infectious node did not infect nor recover, and in the last one, it did not infect and recovers.}
\label{tab:beff}

\end{table}

Note that the above lemma reduces the SIR model to the ICM, as $\gamma=1$ means that each infected node has exactly one opportunity to infect each of its susceptible neighbors.

Let $G=(V,E)$ be a connected undirected network. We will consider 
four centrality measures for groups of nodes $S\subset V$. Namely,
\begin{enumerate}
    \item \textbf{Closeness:} Let $i\not\in S$, and let $d_{iS}\bydef\min_{j\in S}d_{ij}$ be the distance from node $i$ to $S$. Then
    \[
    \bC_{S}\bydef\left[\frac{1}{N-\#S}\sum_{i\not\in S}d_{iS}\right]^{-1}\ ,
    \]
    where $\#S$ is the cardinality of the set $S$.
    The closeness of $S$ is the inverse of the average distance to $S$ of all remaining nodes in the network.

    \item \textbf{Degree:} The degree of a node is defined by the number of first neighbors this node has. For a set $S$
    \[
    \bD_S\bydef\#\{i\not\in S|d_{iS}=1\}\ .
    \]
This is the number of nodes outside $S$ that can be reached from $S$ in a single step.
    \item \textbf{Betweenness:} For all pairs of nodes $i,j\in V$, with $i\ne j$, consider the set of geodesic paths between $i$ and $j$ given by $\mathsf{g}_{(i,j)}$. Then
\[
\bB_{S}\bydef
\sum_{\substack{i,j\in V\\ i\ne j}}\,
\frac{\#\{\ell\in\mathsf{g}_{(i,j)}|\mathop{\mathrm{int}}\ell\cap S\ne\emptyset\}}{\# \mathsf{g}_{(i,j)}}\ .
\]
where the interior of a path $\ell=\{v_0,\cdots,v_k\}$, is given by $\mathop{\mathrm{int}}\ell=\{v_1,\cdots,v_{k-1}\}$. 
The betweenness of $S$ is the sum of the fraction of geodesics that cross $S$ for all pairs $i\ne j$. Note that paths from $i\to j$ and $j\to i$ are considered different. 

    \item \textbf{GED-Walk:} The group exponentially decaying walk, or GED-Walk, was introduced in~\cite{Angriman_2020}, as a generalization of several centrality measures, in particular of betweenness. 
    In this case, walks are considered instead of solely paths or geodesic paths. 
    However, shorter walks contribute more to the centrality measure. Explicitly, let $\phi_\mu(S)$ be the number of walks of size $\mu$ containing at least one node in $S$ and let $0<\alpha<1$ be a fixed parameter. Then
    \[
    \bG_S^\alpha\bydef\sum_{\mu=1}^\infty\alpha^\mu\phi_\mu(S)\ .
    \]
    As the weight of a walk of size $\mu$ in the above summand is $\alpha^\mu$, we will consider in this work $\alpha=\beff$, unless stated otherwise.
\end{enumerate}

\begin{remark}
    There are different definitions of the betweenness centrality measure in the literature. Beyond trivial differences in normalization parameters, some references restrict the geodesics' initial and final nodes to lie outside the set $S$~\cite{Lagos_JGO24}, while others allow extreme points to be in $S$, but a geodesic only counts if it \emph{crosses} the set $S$ (i.e., its interior intersects the set $S$)~\cite{Puzis_PRE2007}. The latter definition is adopted in this work and is also the one implemented in the NetworKit, the modulus of Python used in the numerical simulations~\cite{Staudt}.   
\end{remark}

In Fig.~\ref{fig:Lisboa_Almada} we show a simple example that highlights the advantage of considering group centrality instead of single node centrality; it also helps the 
reader to check the definitions above, as it is not difficult to explicitly calculate all values, with the exception of the GED-Walk. It is a network composed of two cliques linked by a small and a long bridge. The contact points of the two cliques and the two bridges are $e$, $d$, $f$, $j$. 
The two nodes with the highest betweenness are $e$ and $f$, which are the contact points between the small bridge and the cliques. However, when considering groups of two nodes, the pairs $\{e,j\}$ and $\{d,f\}$, which are the two non-adjacent contact nodes in the two bridges, have maximum group betweenness. For other centrality measures, cf.~Fig.~\ref{fig:Lisboa_Almada}.

\begin{figure}
\centering
\includegraphics[width=\textwidth]{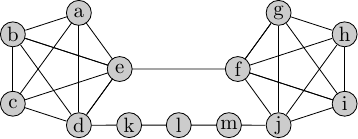}
\caption{Let $x\in \{a,b,c,d\}$ and $y\in \{g,h,i,j\}$. Every geodesic between $x$ and $y$ passes through $e$ and $f$, which are the vertices with the highest betweenness. The sets of two vertices with the highest group betweenness are $\{e,j\}$ and $\{f,d\}$. The same applies to closeness centrality, with $\bC_{e}=\bC_{f}=12/21$, $\bC_{\{d,j\}}=11/12$. For the GED-Walk, with $\alpha=0.5$, the vertices with the highest centrality are $e$ and $f$, while $\{e,j\}$ and $\{f,d\}$ have the highest centrality among sets of two vertices. The vertices with the highest degree are $d$, $e$, $f$, and $j$; but the group of two vertices with the highest degree centrality is $\{d,j\}$. Vertex betweenesses are given by $\bB_{e}=\bB_{f}=48$, $\bB_{d}=\bB_{j}=24$, $\bB_{k}=\bB_{m}=15$, $\bB_{l}=12$ with all the other being equal to zero. The sets of two vertices with highest betweenness are $\bB_{\{e,j\}}=\bB_{\{f,d\}}=69$, followed by $\bB_{\{d,e\}}=\bB_{\{f,j\}}=67$. For comparison $\bB_{\{e,f\}}=58$.
}
\label{fig:Lisboa_Almada}
\end{figure}

In addition to the node and group centrality measures discussed above, it is also worth mentioning the concept of network centralization, a measure introduced in~\cite{Freeman1978} that provides a global assessment of whether the network is built around a central node or is more evenly distributed.

Assume any centrality measure $\chi$, and let $\chi_i$ be its value at node $i\in V$. Assume further that the centrality measure is normalized such that $\max_{i,j\in V}|\chi_i-\chi_j|=1$ and let $\chi^*=\max\{\chi_i|i\in V\}$ be the maximum centrality of any node $i\in V$. Then the centralization of the network is given by
\[
 \chi_{G}=\frac{1}{N-1}\sum_{i\in V}\left[\chi^*-\chi_i\right]\ .
\]

For homogeneous networks, such as the cycle and the complete graph, the network centralization index $\chi_G=0$, whereas for the star network it attains the maximum possible value of 1. These networks will be examined in Sec.~\ref{sec:exact}; other examples, discussed in this and subsequent sections, display centralization values between 0 and 1. Since this concept is constructed based on individual node centrality, it does not play a central role in our analysis and will not be explored in further detail. For a more comprehensive discussion, we refer the reader to~\cite{Koschutzki2005}.

\section{Exact results from simple graphs}
\label{sec:exact}

In this section, we derive exact expressions for the EFS for several simple graphs, for any arbitrary set of vaccinated individuals. Furthermore, in these cases, finding the first three centrality measures for arbitrary sets $V'$ and directly comparing their values with the EFS when $V'$ represents the set of vaccinated individuals is not difficult.  We will show that, for the graphs under consideration, betweenness and GED-Walk are the centrality measures that most effectively identify the groups of individuals to be vaccinated. 

We will use the following notation: $\efs^G_i(N)$ is the expected epidemic final size (EFS) of the network $G=(V,E)$ of size $N=\#V$ when the patient zero occupies node $i$. The expected EFS is calculated under the assumption that the patient zero is uniformly distributed across all nodes in the network and that no individuals have been vaccinated:
\[
\efs^G(N)\bydef\frac{1}{N}\sum_i\efs^G_i(N)\ .
\]

We call $\overline{\efs}^{G,V'}(N)$ the EFS of a graph $G$, when the nodes of set $V'\subset V$ have been vaccinated, in particular $\overline{\efs}^{G,\emptyset}=\efs^G$. All values depend on $\beta$, which will be omitted. 

We will study three types of graphs. The  \emph{path graph}, the \emph{cycle graph}, and the \emph{star graph} (see Fig.~\ref{fig:simple_graphs}). 
We conclude with a remark on the \emph{complete graph}. 

\begin{figure}
\centering
\includegraphics[width=0.98\textwidth]{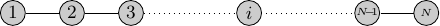}\\
\begin{minipage}{0.48\textwidth}
\includegraphics[width=0.98\textwidth]{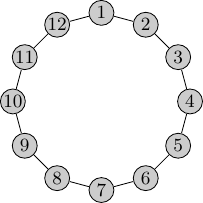}\\
\end{minipage}%
\hfill\begin{minipage}{0.48\textwidth}
\includegraphics[width=0.98\textwidth]{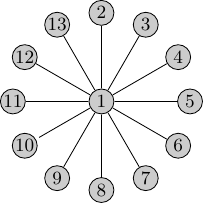}\\
\end{minipage}%
\caption{The path graph $P_N$ (above), the cycle graph $C_N$, with $N=12$ (left, below) and the star graph $S_N$, with $N=13$ (right, below). 
Both the path and the cycle graphs clearly illustrate the differences between two vaccination strategies: selecting the top $k$ individuals with the highest centrality and selecting the group of $k$ individuals with the highest combined centrality. The path graph further highlights the distinctions between using group closeness and group betweenness centrality measures for targeting vaccination groups. The star graph emphasizes the critical role of the central node in disease dynamics.
}
\label{fig:simple_graphs}
\end{figure}

\textbf{Path graph}. The path graph $P_N$ is defined by $V=\{1,\dots,N\}$, and $E=\{\{i,i+1\},i=1,\dots,N-1\}$, i.e., the entries of the adjacency matrix are $a_{i,i\pm 1}=1$ and $a_{ij}=0$, otherwise. 
The EFS, if patient zero is at vertex 1 (one of the extremities), is given by 
\begin{align*}
\efs_1^{\mathrm{path}}(N)&=\sum_{i=1}^{N-1}i\beta^{i-1}(1-\beta)+N\beta^{N-1}\\
&=\frac{1-N\beta^{N-1}+(N-1)\beta^{N}}{1-\beta}+N\beta^{N-1}=\frac{1-\beta^N}{1-\beta}\ .
\end{align*}
For an arbitrary initial condition $1\leq i \leq N$, we have: 
\[
\efs_i^{\mathrm{path}}(N)=\efs_1^{\mathrm{path}}(i)+\efs_1^{\mathrm{path}}(N+1-i)-1=\frac{1+\beta}{1-\beta}-\frac{1}{1-\beta}\left(\beta^i+\beta^{N+1-i}\right)\ .
\]

\begin{remark} 
It is straightforward to derive that 
\[
\efs_{i+1}^{\mathrm{path}}(N)-\efs_i^{\mathrm{path}}(N)= \beta^i - \beta^{N-i} \begin{cases} 
    > 0, & \text{if } ~ i < N/2, \\
    = 0, & \text{if } ~ i  = N/2, \\
    < 0, & \text{if }  ~ i > N/2
\end{cases}
\]
and conclude that 
$\efs_i^{\mathrm{path}}(N)$ is a concave function, and 
\[
\argmax{\efs_i^{\mathrm{path}}(N)} = 
\begin{cases} 
    \lceil \frac{N}{2} \rceil, & \text{if } N \text{ is odd}, \\
    \frac{N}{2} \text{ and } \frac{N}{2} + 1, & \text{if } N \text{ is even}.
\end{cases}
\]
\end{remark}

The expected EFS for the path graph is 
\begin{equation}
\label{eq:severity_line}
\efs^{\mathrm{path}}(N)=\frac{1}{N}\sum_{i=1}^N \efs_i^{\mathrm{path}}(N)=\frac{1+\beta}{1-\beta}-\frac{2(\beta-\beta^{N+1})}{N(1-\beta)^2}\ .
\end{equation}

We now establish the EFS of the path graph $P_N$ when the nodes of an arbitrary subset of nodes  $V'=\{v_1,v_2 , \dots , v_k\}\subset V$, $v_1< v_2 < \dots < v_k$, have been vaccinated.

Define $v_0=0$, $v_{k+1}=N+1$, and let $w_i=v_{i}-v_{i-1}-1$ be the number of nodes in $P_N$ between two consecutive nodes of $V'$. Patient zero appears in any node, with equal probability. If he or she appears at $V'$, then the EFS will be equal to zero. With probability $w_i/N$, it will appear in a segment of size $w_i$. Therefore,
\begin{align}
\nonumber
\overline{\efs}^{\mathrm{path}, V'}\!(N)&\bydef \sum_{i=1}^{k+1}\frac{w_i}{N}\efs^{\mathrm{path}}(w_i)\\
\label{eq:EFShatline}
&=\frac{(N-k)(1+\beta)}{N(1-\beta)}-\frac{2\beta (k+1)}{N(1-\beta)^2}+\frac{2}{N(1-\beta)^2}\sum_{i=1}^{k+1}\beta^{w_i+1}\ .
\end{align}

Minimize the EFS is equivalent to minimize the sum $\sum_{i=1}^{k+1}\beta^{w_i+1}=\sum_{i=1}^{k+1}\beta^{v_{i}-v_{i-1}}$.
Differentiating with respect to $v_i$ and equating to zero, we find $\log\beta\left(\beta^{v_{i}-v_{i-1}}-\beta^{v_{i+1}-v_i}\right)=0$. Therefore 
\begin{equation}
\label{eq:para_v}
v_i^0-v_{i-1}^0=v_{i+1}^0-v_i^0\quad\Longleftrightarrow\quad v_i^0=\frac{v_{i+1}^0+v_{i-1}^0}{2}
\end{equation}
We conclude that to minimize the EFS the vaccinated vertices must be as equally spaced as possible.

We now obtain an explicit formula for the betweenness of $V'$, computing all geodesics (unique unidirectional straight lines) linking given nodes $x\in V$ and $y\in V$ that cross $V'$. We divide the formula into three parts:
\begin{enumerate}
    \item If $x,y\not\in V'$. For fixed $i=1,\dots,k+1$, the number of geodesics starting in the $i$-th component, not finishing in the $i$ component, and that crosses $V'$ is given by $\omega_i\times(N-k-\omega_i)$. The total number of geodesics in this situation is given by
    \[
    \sum_{i=1}^{k+1}\omega_i(N-k-\omega_i)=(N-k)^2-\sum_{i=1}^{k+1}\omega_i^2\ .
    \]
    \item If $x\in V'$, $y\not\in V'$. For fixed $i=1,\dots,k$, the number of geodesics starting at $v_i$, that do not finish in $V'$, and crosses $V'$ is given by $N-k-w_{i}-w_{i+1}$. The total number of geodesics is given by
    \begin{align*}
    2\sum_{i=1}^{k}(N-k-\omega_{i}-\omega_{i+1})&=2k(N-k)-4\sum_{i=1}^{k+1}\omega_i+2\omega_1+2\omega_{k+1}\\
    &=2(k-2)(N-k)+2\omega_1+2\omega_{k+1}\ .
    \end{align*}

    \item If $x,y\in V'$. For $x=v_1$ or $x=v_k$, there are $k-2$ possibilities for $y$; for $x=v_i$, $i=2,\dots,k-1$, there are $k-3$ possibilities for $y$ such that the geodesic between $x$ and $y$ crosses $V'$. Therefore, there are 
    $2(k-2)+(k-2)(k-3)=(k-1)(k-2)$ geodesics in this case.
\end{enumerate}
Adding up the last three equations, we conclude
\begin{align*}
\bB_{V'}&=(N+k-4)(N-k)+(k-1)(k-2)+2\omega_1+2\omega_{k+1}-\sum_{i=1}^{k+1}\omega_i^2\\
&=N(N-2)+k+2(v_1-v_k)-\sum_{i=1}^{k+1}(v_i-v_{i-1}-1)^2\ .
\end{align*}
Differentiating with respect to $v_i$ and equating to 0, we find
\[
v_1=\frac{v_2+1}{2}\ ,\quad v_i=\frac{v_{i-1}+v_{i+1}}{2}\ ,\quad i=2,\dots,k-1\ ,\quad v_k=\frac{v_{k-1}+N}{2}\ .
\]
Therefore, the set that maximizes betweenness is obtained by equally spacing the vaccinated nodes (as much as possible) and assuming fictitious vaccinated individuals at nodes $1$ and $N$, i.e., its distance (measured node by node) is at most one to the set of the same size that, after being vaccinated, minimizes the EFS.

For the closeness, we have
\begin{align*}
(N-k)\bC_{V'}^{-1}&=\frac{v_1(v_1-1)}{2}+\frac{1}{4}\sum_{\begin{subarray}{c} j=1\\ v_{j+1}-v_j\ \text{even}\end{subarray}}^{k-1}(v_{j+1}-v_j)^2\\
&\quad+\frac{1}{4}\sum_{\begin{subarray}{c} j=1\\ v_{j+1}-v_j\ \text{odd}\end{subarray}}^{k-1}\left((v_{j+1}-v_j)^2-1\right)+\frac{(N-v_k+1)(N-v_k)}{2}\\
&=\frac{v_1(v_1-1)}{2}+\frac{1}{4}\sum_{j=1}^{k-1}(v_{j+1}-v_j)^2+\frac{(N-v_k+1)(N-v_k)}{2}-\frac{1}{4}\sum_{\begin{subarray}{c} j=1\\ v_{j+1}-v_j\ \text{odd}\end{subarray}}^{k-1}1\ .
\end{align*}
Due to the summation over odd differences between vaccinated nodes, we could not derive an explicit formula for the maximizing sets. The sum over the square of the distance between vaccinated nodes tends to make the optimal set consist of equally spaced nodes, while the last term tends to make the difference between them odd. This is different from the previous case. In particular, if $N=10$ and $k=2$, closeness is maximized by $\{3,8\}$, while betweenness is maximized by $\{4,7\}$. For $N=20$ and $k=3$, closeness is maximized by $\{3,10,17\}$, $\{4,10,17\}$, $\{4,11,17\}$, and $\{4,11,18\}$, while betweenness is maximized by $\{5,10,15\}$, $\{6,10,15\}$, $\{6,11,15\}$, and $\{6,11,16\}$. Finally, for $N=20$ and $k=4$ the unique maximize of the closeness centrality measure is $\{3,8,13,18\}$, while there are five maximizer for the betweenness: $\{4,8,12,16\}$, $\{5,8,12,16\}$, $\{5,9,12,16\}$, $\{5,9,13,16\}$, and $\{5,9,13,17\}$. Note that while betweenness tends to favor equally spaced nodes, closeness tends to favor equally spaced nodes with odd differences between the position of successive ones.

The maximum degree centrality for subsets of size $k\leq \lfloor N/3\rfloor$ on the path graph $P_N$ is $\bD_{V'}=2 k$. This maximum is achieved for any subset $V’=\{v_1,\dots,v_k\}$ satisfying the constraints: $v_1\geq 2$, $v_i-v_{i-1}\geq 3$ for $i=2,\dots, k$, and $v_k\leq N-1$.  For $k>\lfloor N/3\rfloor$ the maximum degree centrality decreases as $k$ increases. This occurs because the constraints on node selection become tighter, preventing the nodes in the subset  $V’$ from maintaining the optimal spacing required to maximize degree centrality.

Let $\psi_\mu(N)$ be the number of walks of size $\mu$ in a path graph of size $N$. The number of walks that cross any of the nodes $V'=\{v_1,\dots,v_k\}$ is given by 
\[
\phi_\mu(V')=\psi_\mu(N)-\sum_{i=0}^k\psi_\mu(v_{i+1}-v_i-1)
\]
where, again, we used that $v_0=0$ and $v_{k+1}=N+1$. Let $\bG_{V'}^\alpha(N)$ denote the GED-Walk of the subset of nodes $V'\subset\{1,\dots,N\}$ in a path graph of size $N$. Then
\[
\bG_{V'}^{\alpha}(N)=\bG_{\emptyset}^\alpha(N)-\sum_{i=0}^k\bG_\emptyset^\alpha(v_{i+1}-v_i-1)\ .
\]
Let $\mathsf{G}=\bG_\emptyset^\alpha$, and assume that (there is a continuous extension that) $\mathsf{G}$ is differentiable and strictly monotonous.
Differentiating with respect to $v_i$, we find
\[
\frac{\partial\bG_{V'}(N)}{\partial v_i}=-\mathsf{G}'(v_i-v_{i-1}-1)+\mathsf{G}'(v_{i+1}-v_{i}-1)\ .
\]
Equating the previous equation to 0 and from the monotonicity of $\mathsf{G}$, we conclude that 
$\mathsf{G}'(v_i-v_{i-1}-1)=\mathsf{G}'(v_{i+1}-v_{i}-1)$, and the GED-Walk is maximized when (as much as possible) $v_i=\frac{v_{i+1}+v_{i-1}}{2}$.

Therefore, for the path graph, it is evident that betweenness and GED-Walk (among the four centrality measures discussed in Sec.~\ref{sec:basic}) are the most suitable measures for addressing the optimal vaccination problem.

We finish the study of the path graph with an example that clearly shows the difference between optimal spreaders and optimal blockers.

\begin{remark}\label{rmk:spreaders_blockers}
    Consider the path graph $P_3$. The EFS for initially infected nodes at $\{1,2\}$ and $\{2,3\}$ is $2+\beta<2+2\beta-\beta^2$, the EFS when the group of  initially infected nodes is $\{1,3\}$, the optimal spreader. However, the expected EFS is $\frac{1}{3}$ if any group of two nodes is initially vaccinated, implying that all groups of two nodes are optimal blockers.
\end{remark}

\textbf{Cycle graph}:   The cycle graph $C_N$ is defined by $V=\{1,\dots,N\}$ and $E=\{\{i,i+1\},i=1,\dots,N-1\} \cup \{N,1\}$, i.e., the entries of the adjacency matrix are $a_{i,i\pm1}=1$ for $i=1,\dots,N-1$ and $a_{1N}=a_{N1}=1$. We may assume, without loss of generality, that patient zero was at node $1$. Therefore
\begin{align*}
\efs^{\mathrm{cycle}}(N)&=\efs_1^{\mathrm{cycle}}(N)=\sum_{\begin{subarray}{l}m+n\le N-2\\ m,n\ge 0\end{subarray}}(m+n+1)\beta^{m+n}(1-\beta)^2+(1-\beta)N^2\beta^{N-1}+N\beta^N\\
&=(1-\beta)^2\sum_{m=1}^{N-1}m^2\beta^{m-1}+(1-\beta)N^2\beta^{N-1}+N\beta^N\\
&=\frac{1+\beta}{1-\beta}-\frac{\beta^{N-1}}{1-\beta}\left[(N-1)^2\beta^2-(2N^2-2N-1)\beta+N^2\right]\\
&\quad+N^2\beta^{N-1}-N(N-1)\beta^N\ .
\end{align*}

After vaccinating the first node, 
the problem reduces to the path graph $P_{N-1}$, and therefore expressions for the EFS can be obtained from equivalent expressions for the path graph $P_N$. Therefore, the best vaccination strategy for a group of $k>1$ nodes will be to vaccinate as equally spaced as possible. This set maximizes the betweenness and the GED-Walk. The closeness will be also maximized by a similar set, as there are no boundary effects to take into consideration, with the additional constraint that distances between successive nodes should be odd. However, the degree is maximized if no two vaccinated nodes are adjacent. However, maximizing degree centrality does not give a reliable indication of which nodes should be vaccinated to minimize the expected EFS.

\textbf{Star graph}. The star graph $S_N$ 
is defined by a central node $i=1$ connected directly to all other vertices, and those vertices (leaves) have degree 1, i.e., the entries of the adjacency matrix are
$a_{1i}=a_{i1}=1$, for $i=2,\dots,N$ and $a_{ij}=0$, otherwise. 

It is immediate that
\begin{align*}
\efs^{\mathrm{star}}_{1}(N)&=1+(N-1)\beta\ ,\\
\efs^{\mathrm{star}}_{i}(N)&=1+\beta\efs^{\mathrm{star}}_{1}(N-1)=1+\beta+(N-2)\beta^2\ ,\quad i\ne 1\ .
\end{align*}
We conclude
\[
\efs^{\textrm{star}}(N)=\frac{1}{N}\efs_1^{\textrm{star}}(N)+\frac{N-1}{N}\efs^{\textrm{star}}_i(N)=1+\frac{2(N-1)}{N}\beta+\frac{(N-1)(N-2)}{N}\beta^2\\ .
\]

If the set $V'$, with $N-2\ge k=\#V'\ge 1$, is selected to be vaccinated, the expected EFS will be
\begin{align*}
\overline{\efs}^{\mathrm{star},V'}(N)=\left\{
\begin{array}{ll}
&\frac{N-k}{N}\ ,\ \text{if}\ 1\in V'\ ,\\
&\frac{N-k}{N}\efs^{\mathrm{star}}(N-k)>1\ ,\ \text{if}\ 1\not\in V'\ .
\end{array}\right.
\end{align*}
If $k=N-1$, the resulting graph is composed by just one node and the problem is trivial.
Therefore, to minimize the EFS it is necessary to vaccinate the center. On the other hand, it is immediate to check that sets $S_*$ that maximize any given centrality measure will be such that $1\in S_*$. 

For the simple graphs in this section, we showed that the betweenness and GED-Walk stand out as the best proxies of the EFS and, consequently, the optimal vaccination problem. 

Before going to the next section, we finish with one more example:

\textbf{The complete graph}: We could not finish this section without a word on the complete graph, probably the most uninteresting graph, as it is not only homogeneous (i.e., all nodes are equivalent) and isotropic (from any node, all other nodes are equivalent), but also because removing any number of nodes will preserve these two properties. All centrality measures are equivalent and any set of vaccinated nodes will provide the same result. 

Let $p_t^N(s,i)$ be the probability that at time $t$, there are $s$ susceptible and $i$ infectious individuals in a complete graph of size $N$; we say that the population is at state $(s,i)$. We assume that there is one and only one initially infectious individual, i.e., $p_0^N(N-1,1)=1$, $p_0^N(s,i)=0$ otherwise. Therefore
\begin{equation}\label{eq:clique_recursion}
p_{t+1}^N(s,i)=\sum_{k=0}^{N-s-i}p_t^N(s+i,k)\Theta_{s+i,k,i}
\end{equation}
where 
\[
\Theta_{s,i,j}=\binom{s}{j}(1-\beta)^{(s-j)i}\left(1-(1-\beta)^i\right)^j
\]
is the transition probability from state $(s,i)$ to state $(s-j,j)$. Note that $\Theta_{s,i,j}=\mathcal{B}_{j,s}(1-(1-\beta)^i)$, where $\mathcal{B}$ denotes the family of Bernstein polynomials~\cite{GzylPalacios}.

The outbreak will be finished at most at time $t=N$, therefore $\efs^{\mathrm{complete}}(N)\bydef N-\sum_{s=0}^{N-1}sp_N^{N}(s,0)$. Furthermore, $\overline{\efs}^{\mathrm{complete},V'}(N)=\frac{N-\#V'}{N}\efs^{\mathrm{complete}}(N-\#V')$.

It is clear that all sets of the same size are equivalent, and therefore the optimal vaccination problem in the case of the complete graph is trivial, see Table~\ref{table:severity_clique}.

\begin{table}
\centering
\begin{tabular}{c|cp{15cm}}
$N$&$\efs^{\mathrm{complete}}(N)$\\
\hline
$1$& $1$\\
$2$& $\beta+1$\\
$3$& $-2\beta^3 + 2\beta^2 + 2\beta + 1$\\
$4$& $-6\beta^6 + 21\beta^5 - 21\beta^4 + 6\beta^2 + 3\beta + 1$\\
$5$& $24\beta^{10} - 168\beta^9 + 480\beta^8 - 700\beta^7 + 508\beta^6 - 108\beta^5 - 60\beta^4 + 12\beta^3 + 12\beta^2 + 4\beta + 1$\\
\end{tabular}
\caption{EFS for the complete graph of size $N$. If $\beta=1$, $\efs^{\mathrm{complete}}(N)=N$; if $\beta=0$, $\efs^{\mathrm{complete}}(N)=1$.}
\label{table:severity_clique}
\end{table}

\section{Epidemic final size and centrality measures: numerical results}
\label{sec:numerical}

In this section, we consider more general network topologies for which exact results are not possible. Given a network $G=(V,E)$, we assume a single initially infectious individual is located at a randomly chosen node, with all nodes being equally likely, and denote $V'$ as the set of vaccinated individuals. We 
will estimate the expected epidemic final size, $\overline{\efs}^{G,V'}(N)$, analyze its relationship with the number and location of the vaccinated individuals in  $V'$, and explore its correlation with the centrality measures of set $V'$. 

We start recapitulating the procedure introduced in~\cite{Kempe_Tardos_2003} and extensively used in simulations of diffusion processes in social networks. We are not aware of the use of this particular approach in the study of outbreaks in networks and therefore we will briefly explain it. 

Assume that no node was vaccinated (i.e., $V'=\emptyset$), and that each infectious node has exactly one opportunity to infect its neighbors, cf. Lemma~\ref{lem:rescaling}. The central idea of the procedure is to change the focus of the simulation from the nodes to the edges. Consider that each edge acts as a \emph{gate}, which can either be open (transmitting the pathogen) with probability $\beta\in(0,1]$, or closed, with probability $1-\beta$. This will generate an \emph{effective network} $G'=(V,E')$, where $E'\subset E$ is the set of the ``open gate edges''. Each effective network $G'$ can be viewed as a specific realization of the epidemic dynamics with an associated probability of $\beta^{\#E'}(1-\beta)^{\#E-\#E'}$. For each $v\in V$ as the patient zero, the number of nodes reachable from $v$ in the effective network, i.e., the size of the connected component of $G'$ that includes $v$, $\mathcal{C}_{G'}(v)$, is the number of infected individuals in this specific realization of the epidemic spread (see Fig.~\ref{fig:algol_a}).

\begin{figure}
\centering
    \includegraphics[width=0.7\textwidth]{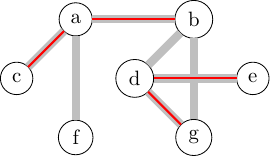}
\caption{Illustration of the EFS resulting from a specific realization of an epidemic dynamic, with transmission probability $\beta$ in a network $G=(V,E)$, with 
$V=\{a,b,c,d,e,f,g\}$ and $E$ the set of the grey edges.
From the original set of edges $E$, a subset of the red edges, $E'$, was selected, creating the effective 
network $G'=(V,E')$, with an associated probability of $\beta^4(1-\beta)^3$. 
If the \emph{patient zero} is located at nodes (a), (b), or (c), the EFS is 3. The same holds for patient zero at 
nodes (d), (e), or (g). If the outbreak starts at node (f), the EFS is 1.}
\label{fig:algol_a}
\end{figure}

The implemented procedure operates as follows. At the end of iteration $k$, $k$ patient zero nodes have been selected (with repetition) from $G$, the effective networks $G'_0, \dots, G'_{\ell}$ have been generated, and the set $F$ contains the $k$ EFS values corresponding to each selected node in the associated effective networks.
In iteration $k+1$, a patient zero node, $v_{k+1}$, is randomly selected from the $N$ nodes of $G$. If $v_{k+1}$ has already served as patient zero in all effective networks $G'_0, \dots, G'_{\ell}$, a new effective network, $G'_{{\ell}+1}$, is generated. Node $v_{k+1}$ is recorded as a patient zero in $G'_{{\ell}+1}$. The connected components of $G'_{{\ell}+1}$ are identified, and the EFS for each node $u \in V$ is set as the size of the component $\mathcal{C}_{G'_{{\ell}+1}}(u)$ to which $u$ belongs. The EFS value for node $v_{k+1}$ is included in $F$.

Alternatively, if $v_{k+1}$ has not yet served as patient zero in some effective network among $G'_0, \dots, G'_{\ell}$, the patient zero record of exactly one of these networks is updated to include $v_{k+1}$, and the size of the component of that network to which $v_{k+1}$ belongs, computed in a previous iteration, is included in $F$. The procedure then proceeds to the next iteration and continues until a target number of iterations, $k_*$, is reached. The average of the $k_*$ values in $F$ is treated as a sample of the EFS. The procedure is repeated $K_*$ times, generating $K_*$ sample values, which are then used to estimate the expected EFS of the network $G$, ${\efs}^{G}(N)$, and to compute the corresponding standard deviation.
In the computational experiments, we set $k_*=1000$ and $K_*=100$.

When applied to networks $G=(V,E)$ that include vaccinated nodes $V'\not=\emptyset$, the procedure is similar but incorporates the following modification. If a vaccinated node $v\in V'$ is selected at the start of new iteration, its EFS value is set to 0, included in $F$, and the algorithm proceeds directly to the next iteration. In all other respects, the algorithm operates on the network $G[V\setminus V']$, which is derived from the original network $G$ by removing the vaccinated nodes $V’$ and all edges incident to them, rather than working directly with $G$.

The proposed algorithm minimizes computational effort by generating new effective networks only when necessary and reusing the same effective network to compute the EFS for all possible initial conditions. Specifically, all the connected components identified to calculate the EFS for a patient zero node selected in one iteration can subsequently be used to determine the EFS for patient zero nodes in future iterations thereby avoiding redundant computations. As a result, the total number of effective networks considered is equal to the maximum number of times any single node is selected as patient zero throughout the algorithm.

In Figs.~\ref{fig:EFS_centrality_SW1} and~\ref{fig:EFS_centrality_Barabasi}, we show how the EFS correlates with several centrality measures in scale-free networks. The group centrality measures were computed using the open-source Python package NetworKit~\cite{Staudt,Angriman2022B}.

Fig.~\ref{fig:EFS_centrality_SW1} considers a ring of $N=100$ nodes, each node is connected to its $\nu=4$ closest neighbors, and for each node and each link, there is a probability $p=0.3$ (rewiring probability) that the link is erased and replaced by a new link, connecting the given node and a randomly selected new node.  This is a particular realization of the Watts-Strogatz network~\cite{Watts_Strogatz_1998}. In Fig.~\ref{fig:EFS_centrality_Barabasi}, we use the Barab\'asi-Albert model, where we steadily increase the number of nodes from $m=1$ to $N=100$ such that each new node is attached to one of the existing nodes, with higher probability to nodes with high number of links (preferential attachment) creating a small number of nodes with high degree (hubs). 

From Figs.~\ref{fig:EFS_centrality_SW1} and~\ref{fig:EFS_centrality_Barabasi}, we observe that centrality measures serve as effective predictors of the epidemic final size. Among these, the GED-Walk centrality and betweenness centrality provide the most reliable results, with GED-Walk being the better predictor. Even though we used a specific functional dependence to fit the data, we do not claim that this choice is robust; that would require further investigation. However, the results clearly demonstrate that the impact of vaccinating any group of individuals can be effectively evaluated based on their centrality, with the most accurate assessments achieved when using the GED-Walk centrality.

\begin{figure}
\centering
    \includegraphics[width=0.48\textwidth]{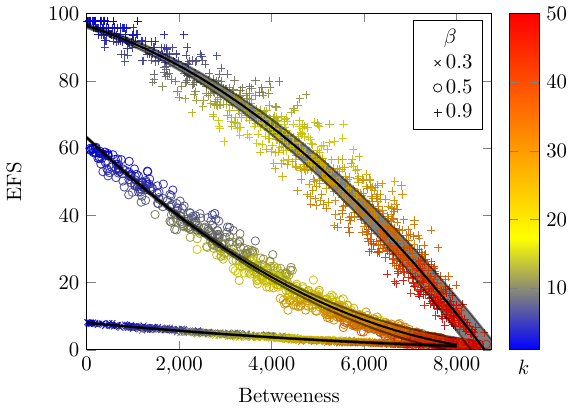}
    \includegraphics[width=0.48\textwidth]{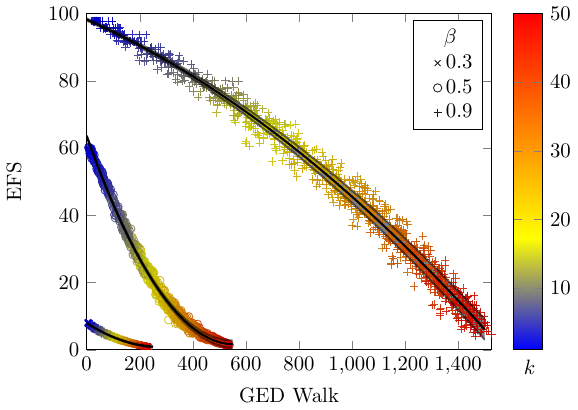}\\
    \includegraphics[width=0.48\textwidth]{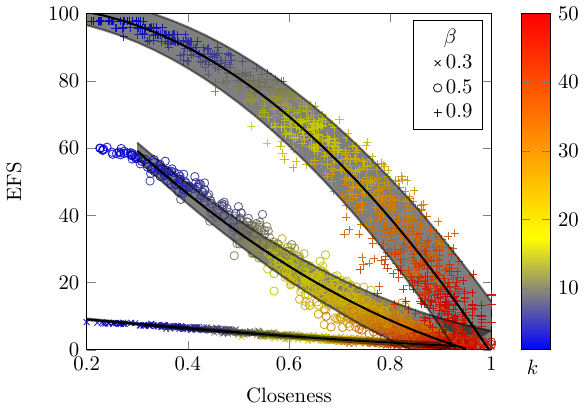}
    \includegraphics[width=0.48\textwidth]{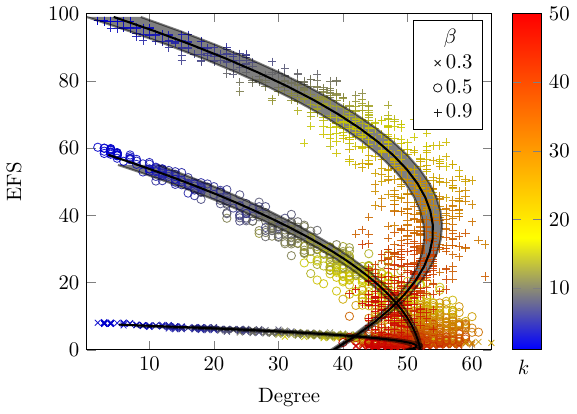}
    \caption{The four panels depict the relationship between the expected epidemic final size (EFS) and various centrality measures for groups of vaccinated individuals on a Watts-Strogatz network with $N=100$ nodes. The network was generated using parameters $\nu=4$ and a rewiring probability $p=0.3$. The results were obtained by considering 1000 groups of vaccinated individuals, with cardinalities $k=1,2,\dots, 50$. The expected EFS, computed for each subset $V'$, is plotted on the vertical axis of the four panels. Markers $\times, \circ, \scriptstyle{+}$ indicate the EFS for transmission probabilities $\beta=0.3, 0.5$, and 0.9, respectively. Consequently, the expected EFS for each vaccinated set $V'$ is displayed three times in each panel, corresponding to the three transmission probabilities. The group centrality measures for each subset $V'$ are plotted on the horizontal axis of the panels as follows: betweenness centrality in the upper left panel, GED-Walk centrality in the upper right panel, closeness centrality in the bottom left panel, and degree centrality in the bottom right panel.
The color coding in the panels represents the cardinality $k$ of the set of vaccinated nodes. 
We observe that the GED-Walk centrality exhibits the highest predictive value, as indicated by its smaller dispersion around the best-fit curve (solid black lines). The best fit was obtained using a quadratic function $f$, $f(x)\to x$ in the case of the Degree, and $x\to f(x)$ in the others. The gray area indicates the error bar in the best-fit quadratic functions. Other functions could give better approximations for a given set of data, but the quadratic function shows a good balance between complexity and precision in all data at the same time, cf. also Fig.~\ref{fig:EFS_centrality_Barabasi}.}
    \label{fig:EFS_centrality_SW1}
\end{figure}

\begin{figure}
\centering
    \includegraphics[width=0.48\textwidth]{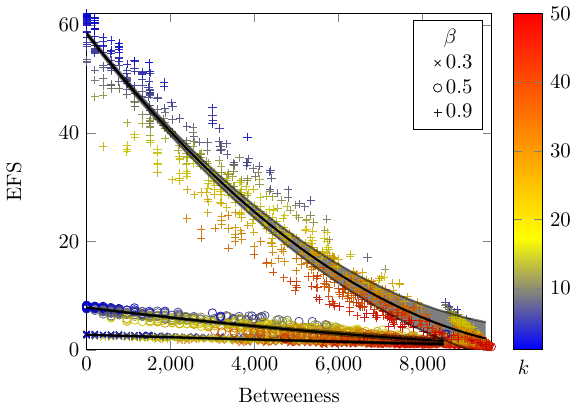}
    \includegraphics[width=0.48\textwidth]{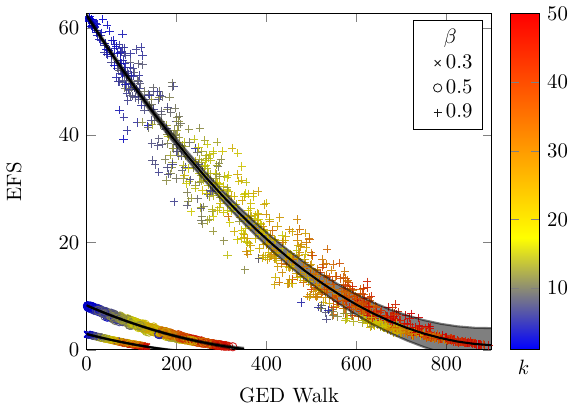}\\
    \includegraphics[width=0.48\textwidth]{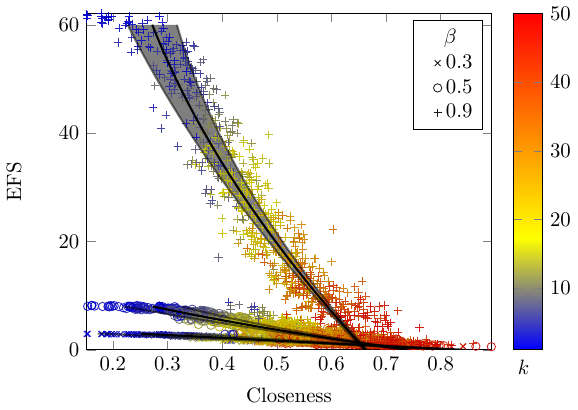}
    \includegraphics[width=0.48\textwidth]{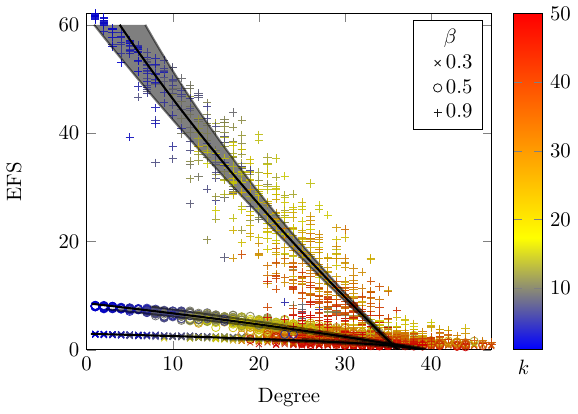}
    \caption{The four panels illustrate the relationship between the expected EFS and various centrality measures for 1000 groups of vaccinated individuals on a Barab\'asi-Albert network with $N=100$ nodes, generated with the parameter $m=1$. The markers, values on the axes, colors, and overall interpretation are consistent with those in Fig.~\ref{fig:EFS_centrality_SW1}. As in Fig.~\ref{fig:EFS_centrality_SW1}, the GED-Walk centrality exhibits the highest predictive value. The best fit was determined using the same functional form as in Fig.~\ref{fig:EFS_centrality_SW1}.
    }
    \label{fig:EFS_centrality_Barabasi}
\end{figure}

\section{Optimal sets}
\label{sec:annealing}
After establishing that centrality measures can serve as proxies for the EFS, this section focuses on selecting optimal subsets of nodes to vaccinate. 
We compare the expected EFS when the $k$ nodes to vaccinate are chosen to maximize each of the four centrality measures under consideration.
To identify the optimal groups of cardinality $k$ with respect to betweenness, GED-Walk, closeness, and degree, we used the Python package NetworKit~\cite{Staudt}.
 These results are then contrasted with the solutions obtained using a simulated annealing algorithm, designed to identify subsets of cardinality $k$ that minimize the expected EFS.

The simulated annealing (SA) algorithm, introduced in~\cite{Kirkpatrick}, is a general heuristic method designed to approximate global optima which is widely used in combinatorial optimization, cf.~\cite{Laarhoven}. At each iteration $i$, the SA algorithm randomly selects a \emph{neighboring} solution $S'$ of the current feasible solution $S$ and probabilistically decides whether to retain $S$ or replace it with $S'$. When minimizing an objective function $\pi$, the change $\Delta_{i}=\pi(S') - \pi(S)$ is computed. If $\Delta_i<0$ (indicating that $S'$ is a better solution), $S'$ replaces $S$ as the current solution. If $\Delta_i\geq 0$,  $S'$ may still be accepted as the new solution with a probability $e^{-\Delta_i/T_i}$, where $T_i>0$ is the temperature at iteration $i$. The temperature $T_i$ decreases gradually to zero according to a predefined cooling schedule.

In our implementation, feasible solutions consist of all subsets of $k$ nodes from the network $G=(V,E)$. A neighbor of a set $S$ is defined as any set obtained by replacing one node in $S$ with a node from $V\setminus S$. For every subset $X\subset V$ of nodes in the network $G$, the objective value $\pi(X)$ is the estimated expected epidemic final size, $\overline{\efs}^{G,X}(N)$, when the nodes in $X$ are vaccinated. This estimate is computed using the procedure described in Section \ref{sec:numerical}. To manage the temperature, we initialize $T_1=1$ and gradually decrease its value by 0.05\% every 11 iterations until it reaches a final temperature of $10^{-3}$. Starting from each initial set of $k$ nodes, the algorithm is run for 1500 iterations, with the best subset of nodes encountered at any stage retained as the final output of that run. The algorithm is executed 11 times using different initial solutions. Ten of these initial solutions consist of sets of $k$ nodes randomly selected from the $N$ nodes in the network. The remaining initial solution is the set of $k$ nodes that, among all those previously considered in our comparative analysis (randomly selected nodes, top $k$ nodes by centrality, or groups of $k$ nodes maximizing group centralities), yielded the lowest estimated EFS when vaccinated. 
The final solution obtained by the SA algorithm is the best subset of $k$ nodes among the solutions from all 11 runs.

When analyzing the results, it is important to note that optimality is not guaranteed for either the group centrality outcomes produced by NetworKit or the solutions obtained using the SA algorithm. 

Fig.~\ref{fig:SimAnneal_a} presents the estimated expected EFS on the Watts-Strogatz network from Fig.~\ref{fig:EFS_centrality_SW1}, for the sets of nodes selected to maximize group betweenness, GED-Walk, closeness and degree centralities, alongside the results  obtained from the SA algorithm. Additionally, it includes the average expected EFS among 20 randomly selected sets of nodes of cardinality $k$ ($k=1,\dots,50)$. These averages are represented by black points in Fig.~\ref{fig:SimAnneal_a}, positioned midway between two horizontal black bars, which indicate the corresponding error bars. The estimates were computed using the procedure described in Section \ref{sec:numerical}.
\begin{figure}
    \centering%
\includegraphics[width=0.33\textwidth]{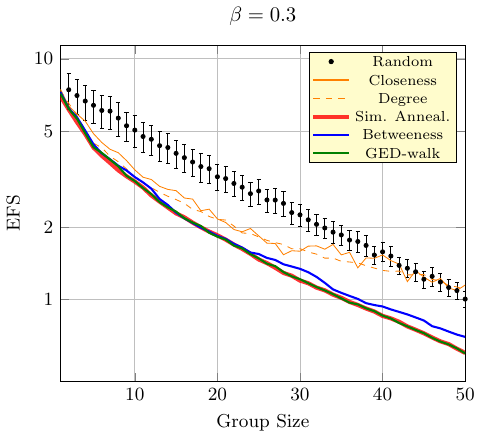}
\includegraphics[width=0.33\textwidth]{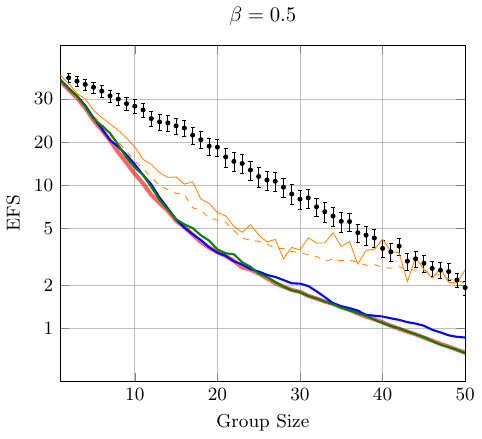}
\includegraphics[width=0.33\textwidth]{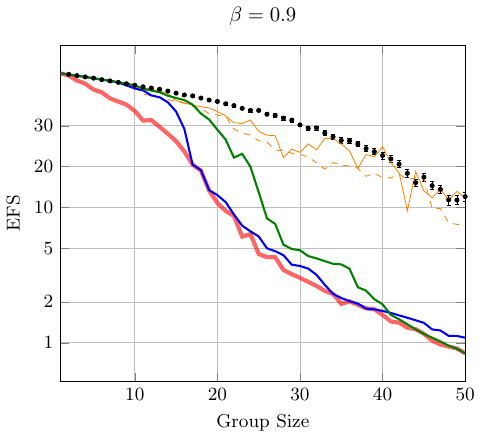}\\
\includegraphics[width=0.33\textwidth]{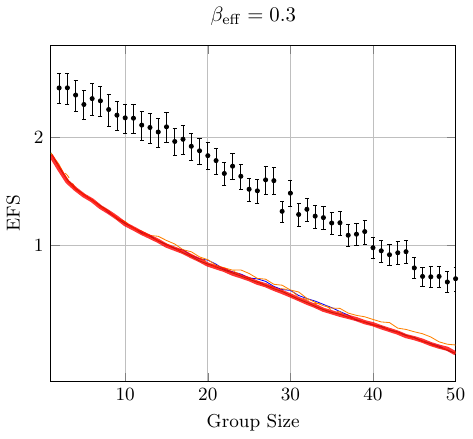}
\includegraphics[width=0.33\textwidth]{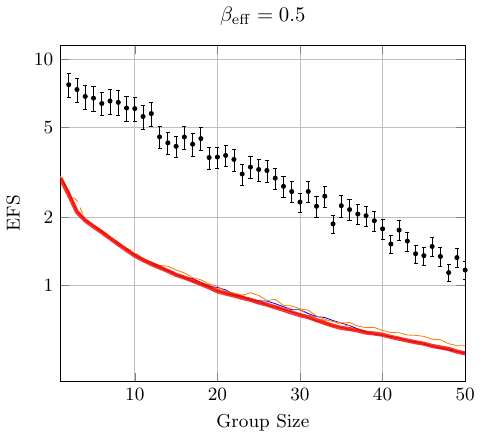}
\includegraphics[width=0.33\textwidth]{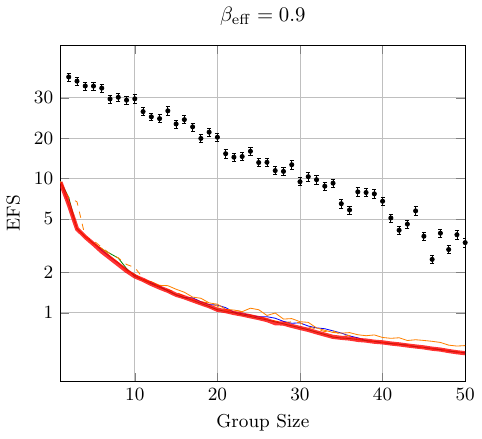}
\caption{
The estimated expected EFS (vertical axis) for varying set sizes (horizontal axis) of vaccinated nodes, including the sets of $k$ nodes that maximize group betweenness, GED-Walk, closeness and degree centralities, as well as those obtained using the SA algorithm. Additionally, it includes the average estimated expected EFS among 20 randomly selected sets of nodes of cardinality $k$. Each average is represented by a back point positioned midway between two horizontal black bars, which indicate the corresponding error bar. These results were obtained for the Watts-Strogatz network from Fig.~\ref{fig:EFS_centrality_SW1} (line above), and for the Barabasi-Albert network from Fig.~\ref{fig:EFS_centrality_Barabasi} (line below), under transmission probabilities $\beta=0.3$ (left), $0.5$ (center) and $0.9$ (right). }
\label{fig:SimAnneal_a}
\end{figure}

The results indicate that the estimated expected EFS for groups maximizing betweenness and GED-Walk centralities are lower than those for groups based on closeness and degree centralities, significantly outperform randomly selected groups, and are comparable to the outcomes achieved using the SA algorithm. However, when the transmission probability $\beta$ is close to 1 and the group size of vaccinees is approximately $1/4$ to $1/3$ of the total population, betweenness emerges as the centrality measure that identifies the most effective group of vaccinees for minimizing the expected EFS. Notably, this range is particularly critical in scenarios with limited vaccination resources.

The results in the first line of Fig.~\ref{fig:SimAnneal_a} were derived using a Watts-Strogatz (small-world) network with $N=100$, $\nu=4$, and a rewiring probability of $p=0.3$. A second realization of the Watts-Strogatz network was not included in this study, as the results were consistent with those presented. Results for the Barab\'asi-Albert graph were included in the second line, showing consistent behavior.

\section{Evaluating method performance on real network data}
\label{sec:USMap}

The studies presented in the previous two sections, which explore the use of group centrality measures as predictors of the EFS in node vaccination strategies, have focused on contact networks that are standard synthetic random networks designed to replicate key structural features commonly observed in real-world systems.

In this section, we will perform a similar study in a real-world network, the United States domestic air transportation system (data available at~\cite{Pajek}, see also~\cite{Wei_2022}).

In this network, nodes are not individuals, and links are not disease transmission routes, but civilian airports and commercial air links, both unweighted. Despite the lack of direct epidemiological interpretation of using this particular network, there are some advantages: (i) There is no explicit information about how the network was constructed, which makes it a neutral test case; thus, the validity of our assertions offers a robust proof of concept for our theoretical framework. (ii) The network allows for straightforward interpretation of results, as the meaning of each node is clear, something that would be far more difficult in a large social media contact network.

In this network, the 332 commercial airports in the U.S. are considered nodes of a network, and links represent the existence of direct flights between two nodes. The results presented here do not take into account the traffic in the links or the importance of each airport, measured by, e.g., the number of passengers. The reasons for not using weighted networks will be discussed in the conclusions. See Fig.~\ref{fig:USMap}.

In Fig.~\ref{fig:EFS_centrality_USAir} we show the correlation between the EFS and the four group centralities used in this work, showing, once more that the GED-Walk is the most reliable one for EFS predictions based on the vaccinated nodes, when compared with the four centrality measures discussed before.

In Fig.~\ref{fig:SimAnneal_USAir} we compare the EFS of the best groups with respect to the four basic centrality measures; the only important difference with respect to the networks studied in Sec.~\ref{sec:numerical} is that the best groups with respect to betweenness and GED-walk have comparable EFS, within numerical errors (not presented). These two measures are close to the optimum value, as validated by the simulated annealing method, and significantly outperform the closeness and degree measures.

One possible reason for that resides in the structure of the network, as a dense air transportation network will provide routes with few steps between almost any destinations. In this sense, best groups for betweenness and GED-walk are similar, reaching a maximum of 90\% of common nodes around $k=30$ (when the similarity between EFS of both optimal groups is maximal) and then stabilizes around 75\% of common nodes when $k$ approaches 147 (circa 50\% of the total number of nodes of the network).

We finish this example by discussing from an intuitive point of view the contribution of each node to the group betweenness. For each size $k$, groups of airports of that size are considered, and the group of $k$ best nodes and the best group of $k$ nodes are determined based on each of the four centrality measures discussed in Sec.~\ref{sec:basic}. When increasing the value of $k$, the
first significant difference between the four centrality measures discussed in the present work is observed at size $k=12$, where the EFS of the set of twelve nodes that maximize the group betweenness is smaller (considering the error bar) than any groups of equal or smaller size that optimize any of the other four centrality measures.
Denver International and Salt Lake City International are among the airports in the selected group of size 12 with high betweenness, yet they do not appear among the top twelve airports when ranked individually by this centrality measure.

The inclusion of Anchorage International Airport (Alaska), the 59th busiest U.S. airport by total passenger boardings, is attributed to the disregard of passenger flow at each node and link. Its centrality comes from the fact that Alaska heavily relies on air transportation, has a large number of airports, and almost all air connections outside Alaska go through Anchorage Intl. Airport.

As a consequence of the limitation of this study to domestic flights, some of the largest U.S. airports, such as Los Angeles Intl. and JFK (New York), 4th and 6th busiest U.S. airports, respectively, were not selected by the centrality measures.

\begin{figure}
    \centering
    \includegraphics[width=1\linewidth]{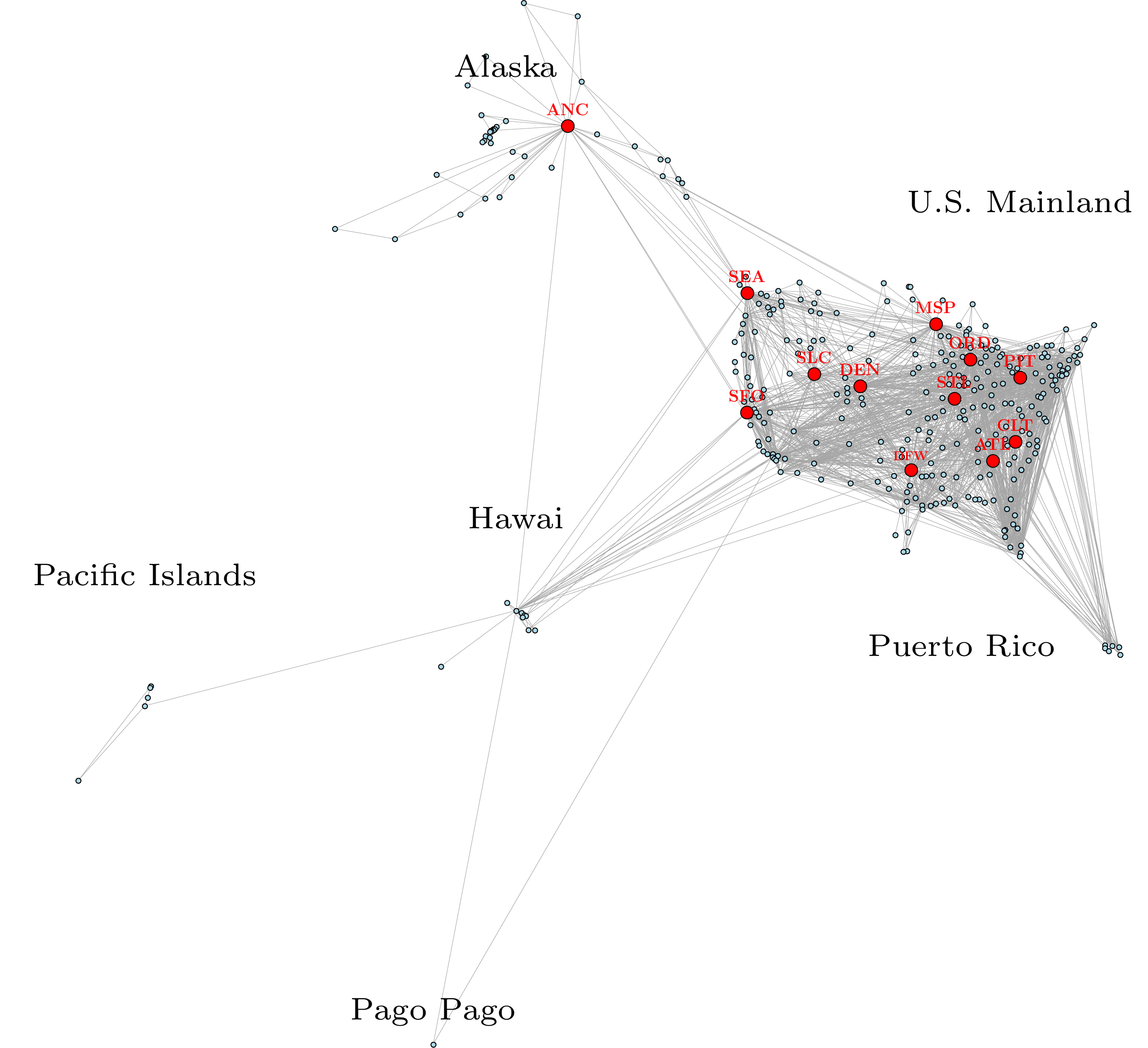}
      \caption{USA map (including overseas territories) with all commercial airports and links indicating the existence of a commercial flight. The group composed of Anchorage Intl. (ANC), Seattle-Tacoma Intl. (SEA), Minneapolis-St Paul Intl. (MSP), Chicago O'hare Intl. (ORD), Salt Lake City Intl. (SLC), Pittsburgh Intll. (PIT), Denver International Airport Stapleton Intl. (DEN), Lambert-St Louis Intl. (STL), San Francisco Intl. (SFO), Charlotte/Douglas Intl. (CLT), The William B. Hartsfield Atlanta (ATL), Dallas/Fort Worth Intl. (DFW) is the group with largest betweenness. For $\beff=0.5$, the EFS is $107\pm 4$, smaller than any other group composed of groups of larger centrality or groups made by nodes of larger centralities, for all centralities studied in the present work and size up to 12.}
    \label{fig:USMap}
\end{figure}

\begin{figure}
 \centering
    \includegraphics[width=0.48\textwidth]{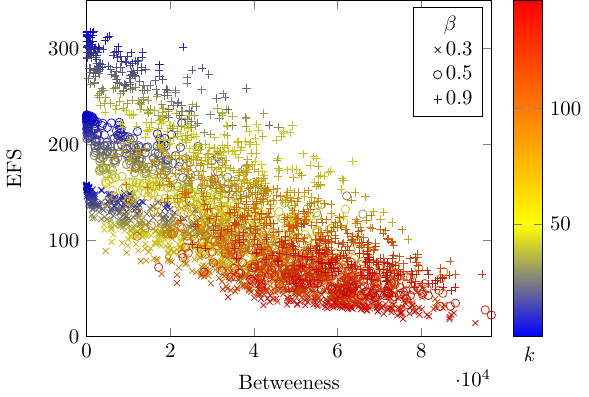}
     \includegraphics[width=0.48\textwidth]{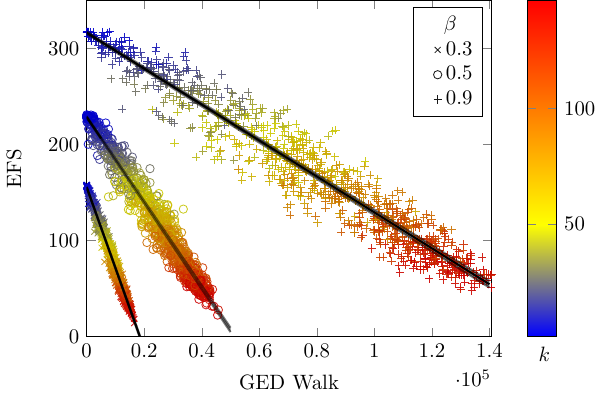}\\
     \includegraphics[width=0.48\textwidth]{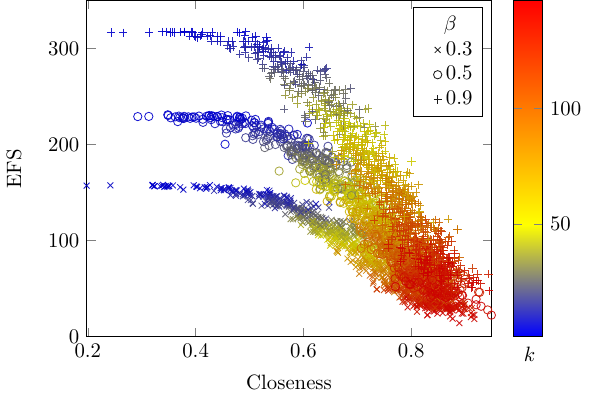}
     \includegraphics[width=0.48\textwidth]{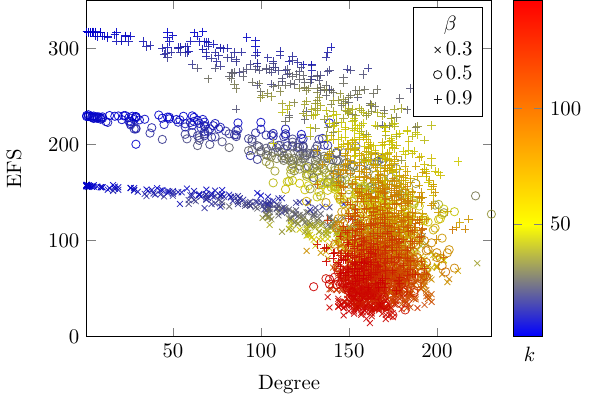}
     \caption{The four panels illustrate the relationship between the expected EFS and various centrality measures for 1000 groups, with cardinalities from size $k=3$ to $k=147$, always in multiples of 3, for the map of the US domestic transportation map discussed in Fig.~\ref{fig:USMap}. The markers, values on the axes, colors, and overall interpretation are consistent with those in Fig.~\ref{fig:EFS_centrality_SW1}. As in Fig.~\ref{fig:EFS_centrality_SW1}, the GED-Walk centrality exhibits the highest predictive value for the EFS. We showed a linear best fit in the GED case and did not provide any specific function for the best fit in the other cases.
     }
    \label{fig:EFS_centrality_USAir}
\end{figure}

\begin{figure}
 \centering
\includegraphics[width=0.33\textwidth]{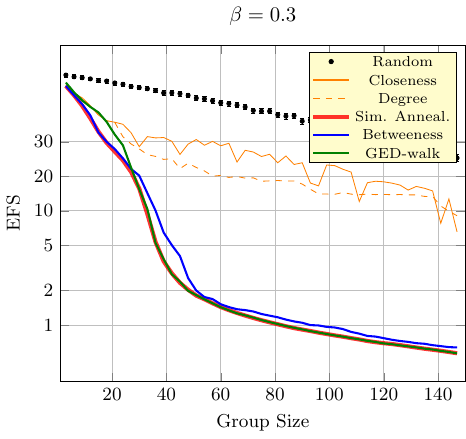}
\includegraphics[width=0.33\textwidth]{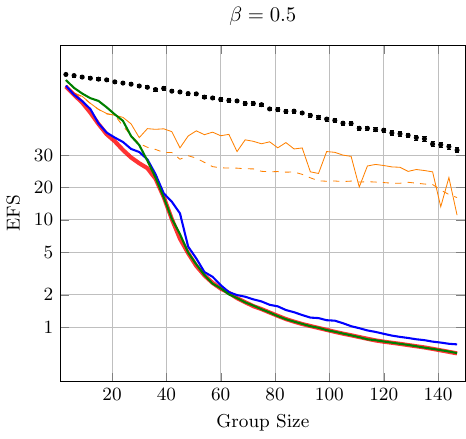}
\includegraphics[width=0.33\textwidth]{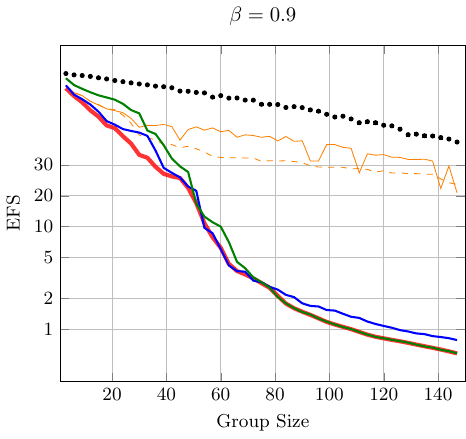}
\caption{
Results are presented in the same format as in Fig.~\ref{fig:SimAnneal_a}, but applied to the U.S. domestic air transportation network shown in Fig.~\ref{fig:USMap}. The findings are consistent: groups of a given size with higher betweenness centrality or GED-Walk centrality act as more effective blockers of epidemic spread when vaccinated. This conclusion is further supported by comparison with the simulated annealing method.
Groups were considered from size $k=3$ to $k=147$, always in multiples of 3.}
\label{fig:SimAnneal_USAir}
\end{figure}

\section{Discussion, limitations, and future work}
\label{sec:discussion}

This research demonstrates the importance of understanding the structure of contact networks in order to grasp how epidemics spread in populations. 
Furthermore, it suggests that centrality measures can serve as valuable predictors of outbreak magnitude across diverse population networks.

Two innovative concepts were introduced in the field of epidemic modeling:

\begin{enumerate}
\item\label{conc:point1} 
We assume and largely accept that public policies do not aim exclusively, or even prioritize, to vaccinate optimal groups, but instead to protect vulnerable individuals even when their behavior has a limited impact on the spread of disease. This is particularly important when vaccines are scarce. Therefore, we need to mathematically model the impact of vaccinating non-optimal subgroups. Our research has shown, using simple epidemic models and networks, that centrality measures can be used to estimate the impact of vaccination. We have found evidence that the GED-Walk is a suitable measure to be used in these situations, which seems to be a novel idea in mathematical epidemiology.
\item\label{conc:point2} 
Building on this, we can provide a fresh perspective on the classical question of optimal vaccination. Specifically, we demonstrate that the betweenness is the best proxy to identify optimal groups. We have tested this claim using the simulated annealing method to optimize (minimize) the final size of the epidemic in some networks. It has become clear that selecting the best group of individuals to be vaccinated is not the same as selecting the best individuals. While this may seem simple, as far as we know, it has never been explicitly stated in the mathematical epidemiology literature. Furthermore, to the best of our knowledge, the simulated annealing algorithm has never been used to validate optimal and near-optimal solutions in optimal vaccination problems.
\end{enumerate}

In more direct words, our main conclusion is that if one wants to find the optimal group of a given size to maximally reduce the outbreak size, one should use the betweenness centrality measure. If one wants to estimate the impact of vaccinating a given group of individuals compared to other groups, neither of them is optimal, then GED-Walk is the most indicated centrality measure. From a topological point of view, this means that close to optimality, the relevance of the geodesics, when compared with all possible paths linking two different nodes, is maximal. In particular, this indicates that optimal vaccination is not concerned with breaking the largest number of transmission routes, but interrupting geodesic ones.

The conclusions presented above need to be validated across a broader range of networks, and their robustness can only be established if they continue to hold under more general epidemic models. We hope to be able to present more examples in the near future.
 
A more general modeling approach that could achieve similar objectives as the method described in point~\ref{conc:point1} is to consider weighted networks. This involves assigning weights to the links (representing a higher or lower probability of disease transmission) and/or the nodes. The latter case, although less commonly found in the literature, would reflect factors such as increased need for medical care, development of severe forms of the disease, or even the probability of death of the individual represented by the node. For the former, we stress that the algorithm presented at Sec.~\ref{sec:numerical} to simulate the EFS can be directly adapted to allow different weights for the links. Weighted networks leads us naturally into the realm of weighted centrality measures, a new and still-developing concept that could be used in the future to address vaccination challenges~\cite{kang2024notiongraphcentralitybased}.

We utilized various classical graphs in our study, not to represent a specific real population, but as a proof of concept regarding the relationship between centrality measures and EFS. However, our epidemic model was limited to cases where immunity, whether from the disease or vaccination, is permanent and does not diminish over time. In addition, an immune individual neither contracts nor transmits the disease. Consequently, a pathogen's path through a population cannot revisit the same node. This implies that a valuable centrality measure could be based on the concept of a self-avoiding walk (SAW), which was originally introduced in the study of polymer growth~\cite{Montroll_JCM1950}; cf. also~\cite{ChalubRiera_JPA1997} for its use in scale-free networks (fractals). However, there is currently no efficient algorithm for obtaining centrality measures based on SAWs. In cases where immunity wanes, it is expected that something in between the GED- and the SAW-based centrality measures will play a prominent role.

Several networks based on real data can be found in~\cite{snapnets}, which may be useful for future studies to further validate our results. Currently, these networks are too large to allow for direct simulation, so new methods need to be developed. However, if the link between group centrality and EFS is validated, estimating the latter based on the former could significantly reduce the need for computational resources.

Another limitation of the model used in the present work is that it does not include pathogen variants. However, this can be easily incorporated. As an extension of our current work, we may also consider different centrality measures, such as a generalization for groups of nodes using the Tukey depth~\cite{Cerdeira_AMC21}, or the ``lobby index,''~\cite{Campiteli_PA2013}; see also~\cite{Berahmand_2018} for another centrality measure.

One important question that was not addressed in this work is the evolution of the network topology due to factors unrelated to the epidemic (external influences) and factors related to the epidemic (e.g., confinements and quarantines) that could be imposed by external agents or undertaken voluntarily. The last point naturally leads to the interaction between epidemic dynamics and human behavior. This scenario could be modeled by activating or deactivating certain links in the network or by varying willingness to be vaccinated. The relationship between epidemic and game theory (the branch of mathematics that models strategic behavior) has been studied since the influential work~\cite{Bauch_2004}, as discussed in the review~\cite{Wang_2016}. Previous works by one of the authors have also examined seasonal epidemics~\cite{Doutoretal_JMB2016} and childhood diseases~\cite{Chalubetal_MB2024} in the realm of voluntary vaccinations.  An alternative modeling approach to a similar problem is to consider two-layered diffusion, one for the pathogen, but a second associated with the diffusion of information, that could be the willingness of a given individual to be vaccinated or the adoption of prophylactic measures (e.g., the use of masks), cf.~\cite{Guo2021,Gross2008}.

In our work, we have not addressed the concept of the basic reproduction number $\mathcal{R}_0$ and its estimation in the simulations presented. Refer to the discussions in~\cite{kiss2017mathematics,DiLauro2020,Pastor_2015} for more information on $\mathcal{R}_0$ in networks. Instead, our focus has been on the parameters $\beta$ and $\gamma$ (or $\beff$). This choice is motivated by several reasons: i) our initial inspiration was the IMP problem, where $\mathcal{R}_0$ is not discussed; ii) $\beff$ is more relevant for our numerical simulations compared to $\mathcal{R}_0$; iii) our findings are not specific to any particular disease; iv) there is no universally accepted translation between the parameters $\beta$, $\gamma$, the network topology, and $\mathcal{R}_0$ that we could leverage.

The networks used in this article represent small populations and should be seen as a starting point for understanding the dynamics and the impact of vaccinations in larger, more realistic networks. When we scale up the number of nodes and/or links, we cannot keep $\beta$ fixed. For instance, in a complete graph with $N$ nodes and a fixed $\beta$, the number of paths connecting two different nodes grows exponentially with $N$. Therefore, for a fixed $\beta$, there is a sufficiently large $N$ for which an initially infected individual will infect the entire network with probability 1. Following the discussion in Rmk.~\ref{rmk:betagamma}, it is natural to consider that the meaning of the parameters $\beta$ and $\gamma$ changes from a finite population, where it is interpreted as a probability to change state, to continuous modeling, where it is more natural to be considered as a rate of change. To find a suitable model in $N\to\infty$ limit, it is necessary to rescale $\beta$ with $N$ in a manner that ensures the limiting model of a complete graph as $N$ approaches infinity is compatible with the differential equation model, as discussed in~\cite{ChalubSouza14,ChalubCorralGarciaRodrigez_2023}. This is why we do not present results on networks of Erd\H os-R\'enyi type nor compare our results with the outbreak dynamics in complete graphs.

 \section*{Acknowledgement}
All the authors are funded by Portuguese national funds through the FCT – Funda\c{c}\~ao para a Ci\^encia e a Tecnologia, I.P., under the scope of the project \hfill UIDB/00297/2020\\ (https://doi.org/10.54499/UIDB/00297/2020) and\hfill UIDP/00297/2020 (https://doi.org/10.54499/UIDP/00297/2020) (Center for Mathematics and Applications --- NOVA Math). All the authors thank Max O. Souza (Universidade Federal Fluminense, Brazil and NOVA Math, Portugal) for interesting discussions on the present work. FACCC also acknowledges the support of the project \emph{Mathematical Modelling of Multi-scale Control Systems: applications to human diseases}  2022.03091.PTDC (https://doi.org/10.54499/2022.03091.PTDC), supported by national funds (OE), through FCT/MCTES. (CoSysM3). Part of this work was done during the stay of FACCC at Carnegie Mellon University (USA), supported by the CMU-Portugal Program, Federal University of Ceará (Brazil), supported by FUNCAPE, and City, University of London (UK).
All the authors contributed equally to the development of the work's ideas, computational codes, data analysis, discussions, and writing of the final version of the manuscript. 

%

\end{document}